\documentclass[preprint,12pt]{elsarticle}

\usepackage{amssymb,times,amsmath,xspace,dsfont,pgf,latexsym,natbib,cmll}
\usepackage[]{graphicx}
\usepackage[all]{xy}

\usepackage{color}
\usepackage{lineno}

\newtheorem{definition}{Definition}[section]
\newtheorem{example}{Example}[section]
\newtheorem{remark}{Remark}[section]
\newtheorem{lemma}{Lemma}[section]
\newtheorem{corollary}{Corollary}[section]
\newtheorem{proposition}{Proposition}[section]
\newtheorem{theorem}{Theorem}[section]

\newenvironment{proof}{\noindent\textbf{Proof. }}{\hfill \small $\Box$}

\begin{document}
	
	\begin{frontmatter}
		
		
		
		\title{Lattice-valued Overlap and Quasi-Overlap Functions}
		
		\author[label1]{Rui Paiva}
		\author[label2]{Eduardo Palmeira}
		\author[label3]{Regivan Santiago}
		\author[label3]{Benjam\'in Bedregal}
		
		\address[label1]{Instituto Federal de Educa\c{c}\~ao, Ci\^{e}ncia e Tecnologia do Cear\'a\\
			Canind\'e, Brazil\\
			Email: rui.brasileiro@ifce.edu.br}
		
		\address[label2]{Departamento de Ci\^{e}ncias Exatas e Tecnol\'{o}gicas \\
			Universidade Estadual de Santa Cruz\\
			Ilh\'eus, Brazil  \\
			Email: espalmeira@uesc.br}
		
		\address[label3]{Departamento de Inform\'atica e Matem\'atica Aplicada  \\
			Universidade Federal do Rio Grande do Norte \\
			Natal, Brazil\\
			Email: regivan@dimap.ufrn.br, bedregal@dimap.ufrn.br}
		
		
		\begin{abstract}
	Overlap functions were introduced as class of bivariate aggregation functions on $[0,1]$ to be applied in image processing. This paper has as main objective to present appropriates definitions of overlap functions considering the scope of lattices and introduced a  more  general  definition, called of quasi-overlaps, which arise  of  abolishes the continuity condition. In addition, are investigated the main properties of (quasi-)overlaps on bounded lattices, namely, convex sum, migrativity, homogeneity, idempotency and cancellation law. Moreover, we make a characterization of Archimedian overlaps.
			
		\end{abstract}

		
	\end{frontmatter}
	
	

          \section{Introduction}

The problem of finding an adjusted way to make a fuzzy partition of a dataset in order to lessen the inaccuracies (overlaping) caused in the decision process regarding the equivalence class a particular data must belong to has been widely studied by researchers through different techniques \cite{YIN17,JIMENEZ,GUO18,ZHANG2018,CHOY19}.

 For instance, in the problem of object recognition what is the best way to avoid overlapping when one wish to classify what is background and what is the object in an image. In this framework, Bustince et al. in \cite{BUSTINCE} introduced the concept of overlap funtion as a possible solution of that problem. According to the authors, those functions provide a mathematical model for this kind of issues where the overlaping degree between functions can be interpreted as the representation of the lack of knowledge between them.

Later, other researchers began to develop deeper studies of overlaps functions and their properties in order to explore their potentialities in different scenarios \cite{BEDREGAL20171,DIMURO2,DIMURO3,GOMEZ201657,QIAO2018,QIAO20181,QIAO20182,DEMIGUEL2018,WANG2018}. From theoretical point of view those papers discuss about some properties of overlap functions and its generalization for intervals and n-dimensional spaces. As an interesting application Bedregal et al. in \cite{BEDREGAL20171} have presented an study about interval image processing by means OWA operators  with interval weights derived from interval-valued overlap functions. 

Recently the lattice theory has increasingly been shown to be a framework for the development of techniques and applications aimed mainly at image processing. Ronse in \cite{RONSE} affirms that for a bounded set of grey-levels, the problem of grey-level overflow can be dealt with correctly only by taking into account the complete lattice structure of the set of grey-level images. Otherwise the properties of morphological operators are lost.

In this paper we propose an extend the concept of overlap  to the lattice context besides studying its main properties according to the main results discussed in the literature. In addition, since the motivation of the continuity of overlaps given in the seminal paper  \cite{BUSTINCE} is not well founded and the role of continuity is quastionable when we consider general and abstract environment as lattice theory, we also introduce a more general notion, namely quasi overlaps on bounded lattices where this condition is abolish.

Section 2 gives a clear review on overlap functions, lattice theory and fuzzy logic operators. Sections 3 and 4 discuss about lattice-valued overlap and quasi overlap functions and its main related properties respectively. Section 5 presents a characterization of Archimedean $L$-overlaps. Finally, in Section 6 some final remarks are considered.

	\section{Preliminaries}

	\noindent
	In this section brings a clear formalization of  key concepts concerning  overlap functions, lattice, homomorphism, retractions and others which are the background of our research. For further reading about them we recommend \cite{BUSTINCE, Birkhoff,COOMAN,priestley, klement,BEDREGAL,PALMEIRA,DIMURO}.
	
	
	\subsection{Overlap functions on $[0,1]$}
	
	
	\noindent
	The notion of overlap function was first introduced by Bustince et al. \cite{BUSTINCE} in order to give a proper characterization of  overlapping in the scenario classification of not crip partition of data. The issue bebind the object recognition problem is find its best classification with respect to background considering the one with less overlapping between the class object and the class background. After that some new theoretical and applied developments have been emerged in the literature regarding these operators  \cite{BEDREGAL20171,DIMURO2,DIMURO3,GOMEZ201657,QIAO2018,QIAO20182,BEDREGAL,Lizasoain,BUSTINCE2}.

	\begin{definition} \cite{BUSTINCE} \label{overlap}
		A mapping $O: [0,1]^2 \rightarrow [0,1]$ is called an overlap function if it satisfies the following conditions:
		\begin{description}
			\item [(O1)] $O(x,y)=O(y,x)$;
			\item [(O2)] $O(x,y)=0$ if and only if $xy=0$;
			\item [(O3)] $O(x,y)=1$ if and only if $xy=1$;
			\item [(O4)] $O$ is non-decreasing;
			\item [(O5)] $O$ is continuous.
		\end{description}
	\end{definition}

	\begin{example}\cite{BEDREGAL} \label{ex1}
		The mapping given by $O_{\min}(x,y)=min\{x,y\}$, $O_P(x,y)=xy$ and $O_{\min \max}(x,y) = \min (x,y) \max (x^2,y^2)$ , for all $x, y \in [0,1]$, are examples of overlaps functions. Moreover, if $O: [0,1]^2 \rightarrow [0,1]$ is an overlap function then 
		$O^2(x,y)=O(x^2,y^2)$ and $O^{\sqrt{}}(x,y)=O(\sqrt{x},\sqrt{y})$ are also overlap functions.
	\end{example}
	
\begin{remark}	
	Given two different overlap functions $O_1$ and $O_2$ then it is possible to obtain  some other interesting examples of overlaps are as follows:
	\begin{enumerate}
		\item $(O_1 \wedge O_2)(x,y) = \min (O_1(x,y),O_2(x,y))$;
		\item $(O_1 \vee O_2)(x,y) = \max (O_1(x,y),O_2(x,y))$;
		\item $O(x,y) = w O_1(x,y) + (1-w) O_2(x,y)$, for each $w \in [0,1]$.
	\end{enumerate}
\end{remark}
	
	\begin{definition}
		Consider $\alpha \in [0,1]$. A bivariate operation $O:[0,1]^2 \rightarrow [0,1]$ is called an $\alpha$-migrative if  
		\begin{equation}\label{eq1}
		O(\alpha x, y) = O(x, \alpha y), \, \text{for all} \, \ x, y \in [0,1].
		\end{equation}
		In case $O$ is $\alpha$-migrative for all $\alpha \in [0,1]$ then $O$ is simply called migrative.
	\end{definition}
	 \begin{example}
Overlap function $O_{\min}$ and $O_P$ (as defined in Example \ref{ex1}) are example of migrative overlaps  (for every $\alpha \in [0,1]$). 
\end{example}
	
%
%

Also, recall that a function $A:[0,1]^n \rightarrow [0,1]$ is called an (n-ary) aggregation function if it is nondecreasing and satisfies the boundary conditions $A(0,\ldots,0)=0$ and $A(1,\ldots,1)=1$. 
	Some other properties related to aggregation functions are listed below:
	\begin{itemize}
		\item [(A1)] An element $a \in [0,1]$ is called  an annihilator of $A$  if $A(x_1, \ldots, x_n)=a$ whenever $a \in \{x_1, \ldots, x_n\}$;
		\item [(A2)] $A$ is said to be strictly increasing if it is strictly;
		\item [(A3)] $A$ is said to have divisors of zero if there exist $x_1, \ldots, x_n \in ]0,1]$ such that $A(x_1, \ldots, x_n)=0$;
		\item  [(A4)] $A$ is said to be idempotent if $M(x, \ldots, x)=x$ for any $x \in [0,1]$.
	\end{itemize}
Notice that there is a way to generalize this property by considering a binary aggregation function $A$ and rewriting the Equation (\ref{eq1}) as follows:
	\begin{equation}\label{eq2}
	O(A(\alpha,x), y) = O(x, A(\alpha,y)), \, \text{for all} \, \ x, y \in [0,1].
	\end{equation}
	as proved by Bustince et al. in \cite{BUSTINCE2}. In this case, we say that $O$ is ($\alpha$,A)-migrative and just A-migrative if $O$ is ($\alpha$,A)-migrative for all $\alpha \in [0,1]$.


\subsection{Bounded lattices: definition and related concepts}


	\noindent
	
	\begin{definition} \label{def_lat} \cite{Birkhoff}
		Let $L$ be a nonempty set. If $\wedge_L$ and $\vee_L$ are two binary operations on $L$, then $\langle L,\wedge_L, \vee_L \rangle$ is called a lattice provided that for each $x, y, z \in L$, the following properties hold:
		\begin{enumerate}
			\item $x \wedge_L y = y \wedge_L x$ and $x \vee_L y = y \vee_L x$ (symmetry);
			\item $(x \wedge_L y) \wedge_L z = x \wedge_L (y \wedge_L z)$ and $(x \wedge_L y)
			\vee_L z = x \vee_L (y \wedge_L z)$ (associativity);
			\item $x \wedge_L (x \vee_L y) = x$
			and $x \vee_L (x \wedge_L y) = x$ (distributivity).
		\end{enumerate}
If in $\langle L,\wedge_L, \vee_L \rangle$ there are elements $0_L$ and $1_L$ such that, for all $x \in L$, $x \wedge_L 1_L = x$ and $x \vee_L 0_L = x$, then $\langle L,\wedge_L, \vee_L, 0_L, 1_L \rangle$ is called a bounded lattice. Also, $L$ is called a complete lattice if every subset of it has a supremum and an infimum element.
	\end{definition}

	Recall that given a lattice $L$ relation
	\begin{equation}
	x \leqslant_L y \ \ \text{if and only if} \, \, x \wedge_L y = x
	\end{equation}
	defines a partial order on $L$. This order will be used by us to compare
	elements.

	\begin{example} \label{ex_1}
		The set $[0,1]$ endowed with the operations defined by \linebreak $x \wedge y = \min \{x,y\}$ and $x \vee y = \max \{x,y\}$
		for all $x, y \in [0,1]$ is a (complete) bounded lattice in the sense of Definition \ref{def_lat} which has $0$ as
		the bottom and $1$ as the top element.
	\end{example}
	
	
\begin{remark}
		When $\leq_L$ is a partial order on $L$ and there are at least two elements $x$ and $y$ belonging to $L$ such that neither
		$x \leq_L y$ nor $y \leq_L x$. In this case, these elements are said to be incomparable and we denoted  by  $x\parallel y$. 
	\end{remark}
	
	
\begin{definition}
		Let $(L,\wedge_L, \vee_L,0_L,1_L)$ and $(M,\wedge_M, \vee_M,0_M,1_M)$ be bounded lattices. A mapping $f:L \rightarrow M$ is called a  lattice homomorphism if for all $x, y \in L$ we have
		\begin{enumerate}
			\item $f(x \wedge_L y) = f(x) \wedge_M f(y)$;
			\item $f(x \vee_L y) = f(x) \vee_M f(y)$;
		\end{enumerate}
		\label{alghom}
In case $f$ is such that $f(x)=0_M$ and $f(y)=1_M$ if and only if $x=0_L$ and $y=1_L$ it is called an $\{0,1\}$-homomorphism.
	\end{definition}
	
	\begin{remark}
		An injective (a surjective) lattice homomorphism  is called a mo\-no\-morphism (epimorphism) and a bijective lattice homomorphism is called an isomorphism. An automorphism is an isomorphism from a lattice onto itself. 
	\end{remark}
	
	\vspace{0.2cm}
	
	\begin{proposition} \cite{PALMEIRA2}
		Every lattice homomorphism preserves the order. \label{prop5}
	\end{proposition}
	
	\vspace{0.2cm}
	
	
	\begin{proposition} \label{pro-aut-char} \cite{PALMEIRA}
		Let $L$ be a bounded lattice. Then a function $\rho: L\rightarrow L$ is an automorphism if and only if
		\begin{enumerate}
			\item $\rho$ is bijective and
			\item $x\leqslant_L y$ if and only if  $\rho(x)\leqslant_L \rho(y)$.
		\end{enumerate}
	\end{proposition}
	
	\vspace{0.2cm}
	
		
	
\begin{remark}
	From now on, lattice homomorphisms will be called just homomorphisms for simplicity.
\end{remark}
	
	
	\begin{proposition}\label{automorfismo}
		Let $\rho: L \to L$ be a automorphism and $f:L^n \to L$ be a function. If $\rho\left(f(x_1, \ldots , x_n)\right)=f\left(\rho(x_1), \ldots ,\rho(x_n)\right)$ then 
\begin{equation} \label{auto_comut}
\rho^{-1}(f(x_1, \ldots ,x_n))=f\left(\rho^{-1}(x_1), \ldots , \rho^{-1}(x_n)\right).
\end{equation}
	\end{proposition}
	\begin{proof}
		Since $f\left(\rho^{-1}(x), \ldots, \rho^{-1}(x_n)\right) = \rho ^{-1}\left(\rho \left(f\left(\rho^{-1}(x), \ldots  , \rho^{-1}(x_n)\right)\right) \right)$ it follows, by hypothesis that $$\rho ^{-1} \left(f\left(\rho\left(\rho^{-1}(x_1)\right), \ldots , \rho \left(\rho^{-1}(x_n)\right)\right) \right)= \rho^{-1} \left(f(x_1, \ldots ,x_n)\right)$$ 	\end{proof}
	\begin{definition}
		Given a function $f:L^n\rightarrow L$, the \textit{action of  an $L$-automorphism} $\rho$ over $f$ results in the function
		$f^\rho:L^n\rightarrow L$ defined as
		
		\begin{equation}
		\label{eq-act-aut}
		f^\rho(x_1,\ldots,x_n)=\rho^{-1}(f(\rho(x_1),\ldots,\rho(x_n)))
		\end{equation}
		In this case, $f^\rho$ is called  the \textit{conjugate} of $f$ (see \cite{Costa}).
	\end{definition}

	\subsection{Limit and Continuity} \label{cont}
	

In this section we discuss about continuity of lattice-valued function and its properties. For a  deeper reading we recommend \cite{SCOTT}.

First, notice that if  $J$ is an index set and $L$ is a lattice then a net in $L$ is defined as a function that associate each $i \in J$ to an element $x_i \in L$ i.e. $i \mapsto x_i: J \rightarrow L$ and  denoted by $(x_i)_{i \in J}$. The limit of a net on $L$ is defined as follows.

	\begin{definition} \label{def_lim1}  \cite{klement}
		Let $L$ be a complete lattice. For any net $(x_i)_{i \in J}$ we write 
		\begin{equation}
		\underline{\lim}_{i\in J} x_i= \sup_{i\in J} \inf_{j \geq i} x_j 
		\end{equation}
		and call $\underline{\lim}_{i\in J} x_i$ the lower limit. Similarly, 
		\begin{equation}
		\overline{\lim}_{i\in J} x_i= \inf_{i\in J} \sup_{j \geq i} x_j 
		\end{equation}
		is called the upper limit. 
		Let $S$ be the class of those elements $\underline{x} \in L$ such that $\underline{x} \leq_L \underline{\lim}_{i\in J} x_i$ and $T$ be the class of those elements $\overline{x} \in L$ such that $\overline{\lim}_{i\in J} x_i\leq \overline{x}$, both for the 
		net $(x_i)_{i \in J}$. For each such elements we say that $x_1$ is a lower $S$-limit and $x_2$ is a upper $T$-limit of $(x_i)_{i \in J}$. If $x_1 = \sup_{\underline{x} \in S} \underline{x}$ and $x_2 = \inf_{\overline{x} \in T} \overline{x}$, we write respectively  $\underline{x} \equiv_S \underline{\lim} x_i$ and $\overline{x} \equiv_T \overline{\lim} x_i$. 
		
	\end{definition}
	
	Alternatively, the notion of lower $S$-limit ($T$-limit) of a net $(x_i)_{i \in J}$ can be defined by means of directed (filtered) sets \cite{SCOTT}. 

\begin{definition}
Let be $X$  a poset\footnote{A poset is a nonempty set $X$ equipped with a partial order $\leq$.}. A subset $D \subseteq X$ is called a directed (filtered) set if every two-element subset of D has an upper (lower) bound in D. A poset is a directed complete partial order (DCPO) if every directed subset has a supremum, and it is a filtered complete partial order (FCPO) if every filtered subset has an infimum. 
\end{definition}
	
	\begin{definition}[\cite{SCOTT}] \label{def_lim2}
		A point $y \in L$ is an eventual lower bound of a net $(x_i)_{i \in J}$ if there exists a $k \in J$ such that $y \leq_L x_i$ for all $i \geq k$. Let $S$ be the class of those pairs  $((x_i)_{i \in J}, \underline{x})$ such that $\underline{x} \leq_L \sup D$ for some directed set $D$ of eventual lower bound of net $(x_i)_{i \in J}$. For each such pair we say that $\underline{x}$ is a lower $S$-limit of $(x_i)_{i \in J}$ and write $ \underline{x} \equiv_S \underline{\lim} x_i$.
	\end{definition}

 Dually, a point $y \in L$ is an eventual upper bound of a net $(x_i)_{i \in J}$ if there exists a $k \in J$ such that $x_i\leq_L y$ for all $i \geq k$. Let $T$ be the class of those pairs $((x_i)_{i \in J}, \overline{x})$ such that $\inf D \leq_L \overline{x}$ for some filtered set $D$ of eventual upper bound of net $(x_i)_{i \in J}$. For each such pair we say that $\overline{x}$ is a upper $T$-limit of $(x_i)_{i \in J}$ and write $ \overline{x} \equiv_T \overline{\lim} x_i$.
	
	Notice that Definition \ref{def_lim2} agrees with Definition \ref{def_lim1} (see \cite[p. 133]{gierz}) for complete lattices. 

	
	\begin{proposition} \cite[prop. II-2.1]{gierz} \label{lat_cont}
		Let $L$ and $M$ be DCPO's and $f:L \rightarrow M$  a function. The following conditions are equivalent:
		\begin{enumerate}
			\item $f$ preserves suprema of directed sets, i.e. $f$ is order preserving and 
			\begin{equation} \label{scott_cont}
			f(\sup \Delta) = \sup \{f(x) \ | \ x \in \Delta\}
			\end{equation}
			for all directed subset $\Delta$ of $L$;
			\item $f$ is order preserving and 
			\begin{equation}
			f(\underline{\lim}_{i\in J} x_i) \leq_L \underline{\lim}_{i\in J} f(x_i)
			\end{equation}
			for any net $(x_i)_{i \in J}$ on $L$ such that $\underline{\lim}_{i\in J} x_i$ and $\underline{\lim}_{i\in J} f(x_i)$ both exist.
		\end{enumerate}
	\end{proposition}
	
	Similarly, the dual proposition can be demonstrated.
	
	\begin{proposition}\label{lat_cont dual}
		Let $L$ and $M$ be FCPO's and $f:L \rightarrow M$  a function. The following conditions are equivalent:
		\begin{enumerate}
			\item $f$ preserves infimum of filtered sets, i.e. $f$ is order preserving and 
			\begin{equation} \label{scott_cont dual}
			f(\inf \Delta) = \inf \{f(x) \ | \ x \in \Delta\}
			\end{equation}
			for all filtered subset $\Delta$ of $L$;
			\item $f$ is order preserving and 
			\begin{equation}
			f(\overline{\lim}_{i\in J} x_i) \geq_L \overline{\lim}_{i\in J} f(x_i)
			\end{equation}
			for any net $(x_i)_{i \in J}$ on $L$ such that $\overline{\lim}_{i\in J} x_i$ and $\overline{\lim}_{i\in J} f(x_i)$ both exist.
		\end{enumerate}
	\end{proposition}
	
	Notice that all complete lattice is a DCPO (FCPO) in which $\underline{\lim}_{i\in J} x_i$ and $\underline{\lim}_{i\in J} f(x_i)$ ($\overline{\lim}_{i\in J} x_i$ and $\overline{\lim}_{i\in J} f(x_i)$) always exist \cite{gierz}. Hence Propositions \ref{lat_cont} and \ref{lat_cont dual} hold for complete lattices. 
	
	Let $\langle L, \leq \rangle$ be a poset. Recall that $U\subseteq L$ is an \textit{upper set} if for every $x,y \in L$ if $x \in U$ and $x\leq y$, then $y\in U$. Also recall that $U$ is a \textit{down set} if for every $x,y \in L$ if $x \in U$ and $y\leq x$, then $y\in U$. A set $X\subseteq L$ is called Scott-open if it is an upper set and if all directed sets $D$ with supremum in $X$ have non-empty intersection with $X$. The Scott-open subsets of $L$ form a topology on $L$, the Scott topology.

There is a connection between convergence given in order theoretic terms by lower limits, or liminfs and Scott topology. In this perspective the Equations (\ref{scott_cont}) and (\ref{scott_cont dual}) generalize the notions of left and right continuity for the unit interval $[0,1]$. This fact motivate the following definition.

\begin{definition}\label{continuous}
		Let $L$ and $M$ be two complete lattices. A function $f:L \rightarrow M$ is called lattice left (right) continuous if and only if it satisfies Equation (\ref{scott_cont}) ( Equation (\ref{scott_cont dual})). In case $f$ is both a left and right continuous function, then $f$ is called a continuous function.
	\end{definition}
	
	\begin{remark} \label{rem-cont-fin-L} When $L$ is finite any function $f:L \rightarrow M$ is continuous because for each directed set $\Delta$ of $L$, $\sup \Delta\in \Delta$ and for each filtered set $\Delta$, $\inf \Delta\in \Delta$.
	      \end{remark}


\begin{definition}\label{limitelateral}
Let $f$ be a bivariate function on a complete lattice $L$ and $a,b \in L$. We will write $\lim\limits_{x \to a^+}f(x,x)=b$ if and only if for every net $(x_i)_{i\in J}$ in $L$ such that $a$ is an eventual lower bounded of $(x_i)_{i\in J}$ and $\lim_{i\in J} x_i=a$ implies $\lim_{i\in J} f(x_i,x_i)=b$. Moreover, for a special case of nets  $(a_n)_{n \in \mathbb{N}}$ in $L$ we define $\lim\limits_{n \to \infty} a_n =a$ if, and only if, exist $k\in \mathbb{N}$ such that $a_n \leq_L a$, for all $n \geq k$, where $a\in L$ is an eventual upper bound of a net $(a_n)_{n \in \mathbb{N}}$.
\end{definition}

	
	\subsection{T-norms and T-conorms on $L$}
	

	\noindent
	It presented here a short formalization for the notion of t-norms and t-conorms on bounded lattices as well as some particular results used in this work. For a deeper view on them we recommend \cite{PALMEIRA,PALMEIRA2}.
	
	\begin{definition}
		Let $L$ be a bounded lattice. A binary operation $T: L^2 \rightarrow L$ is called a $t$-norm if, for all $x, y, z \in L$, it satisfies:
		\begin{enumerate}
			\item $T(x,y)=T(y,x)$ (commutativity);
			\item $T(x,T(y,z)) = T(T(x,y),z)$ (associativity);
			\item If $x \leqslant_L y$ then $T(x,z) \leqslant_L T(y,z), \ \forall \ z \in
			L$ (monotonicity);
			\item $T(x,1_L)= x$ (boundary condition).
		\end{enumerate}\label{tn}
	\end{definition}
	
	\begin{example}
		Let $L$ be a bounded lattice. Thus,  the function $T:L^2 \rightarrow L$ defined by $T(x,y)=x \wedge_L y$ is a t-norm that generalize the classical fuzzy t-norm of minimum, i.e. $T_M:[0,1]^2 \rightarrow [0,1]$ such that $T_M(x,y)= \min \{x,y\}$ for all $x,y \in [0,1]$.
	\end{example}

\begin{definition}
A t-norm $T$ on a lattice $L$ is called \begin{enumerate}
	\item [(i)] $\wedge$-distributive if \,  $T(x,y\wedge z)=T(x,y)\wedge T(x,z)$ for all $x,y, z\in L$
	\item [(ii)] $\vee$-distributive if \, $T(x,y\vee z)=T(x,y)\vee T(x,z)$ for all $x,y, z\in L$
\end{enumerate} If the items (i) and (ii) are both satisfied, then $T$ is called ($\wedge$ and $\vee$)-distributive. \end{definition}

The following definition provides a condition for an element $y$ of a lattice $L$ to belong to the image of the unary operation $T(x, \cdot ): L \to L$.

	\begin{definition}[\cite{SAMINGERPLATZ}]\label{divisible}
A lattice $L$ equipped with some t-norm $T:L^2 \rightarrow L$ is called \textit{divisible} if for all $x,y \in L$ with $y\leq_L x$ there exists some $z\in L$ such that $y=T(x,z)$.
	\end{definition}

	Dually, it is possible to define the concept of t-conorms.
	
	\begin{definition}\label{tcn}
		Let $L$ be a bounded lattice. A binary operation $S: L^2 \rightarrow L$ is called a t-conorm if, for all $x, y, z \in L$, we have:
		\begin{enumerate}
			\item $S(x,y)=S(y,x)$ (commutativity);
			\item $S(x,S(y,z)) = S(S(x,y),z)$ (associativity);
			\item If $x \leq y$ then $S(x,z) \leq S(y,z), \ \forall \ z \in
			L$ (monotonicity);
			\item $S(x,0_L)= x$ (boundary condition).
		\end{enumerate}
	\end{definition}
	
	Notice that $T(x,y) \leqslant_L x$ (or $T(x,y) \leqslant_L y$) and $x \leqslant_L S(x,y)$ (or $y \leqslant_L S(x,y)$) for all $x, y \in L$. In fact, $T(x,y) \leqslant_L x \wedge y \leqslant_L x$ and $x \leqslant_L x \vee_L y \leqslant_L S(x,y)$.
	
	\begin{example}
		Given an arbitrary bounded lattice $L$, the function $S$ given by \linebreak $S(x,y)=x \vee_L y$ for all $x,y \in L$ is a t-conorm
		on $L$ that generalize the classical fuzzy t-conorm of maximum, i.e. $S_M(x,y)= \max \{x,y\}$ for all $x,y \in L$.
	\end{example}
	
	\begin{proposition}\cite[Corollary 2]{BBB13}
		Let $\rho$ be an automorphism on $L$. A t-conorm $S: L^2 \rightarrow L$ satisfies
		\begin{equation} \label{positive_tconorm}
		S(x,y) = 1_L \ \text{if and only if} \ x=1_L \ \text{or} \ y=1_L
		\end{equation}
		if and only if $S^{\rho}$ satisfies also it. A t-conorm satisfying (\ref{positive_tconorm}) is called positive.
	\end{proposition}

%
%
	
	Similarly, it can be proved the following.

	\begin{proposition}
		Let $\rho$ be an automorphism on $L$. A t-norm $T: L^2 \rightarrow L$ satisfies
		\begin{equation} \label{positive_tnorm}
		T(x,y) = 0_L \ \text{if and only if} \ x=0_L \ \text{or} \ y=0_L
		\end{equation}
		if and only if $T^{\rho}$ satisfies also it. A t-norm satisfying (\ref{positive_tnorm}) is called positive.
	\end{proposition}

	\section{Overlaps and quasi-overlaps on bounded lattices}


Overlap functions were designed from the attempt to solve the problem of
imprecision in the image classification process as explains Bustince et al. \cite{BUSTINCE}. Authors further state that overlap functions were first defined for $[0,1]$ but that other domain could be naturally considered. Thus, based on this assumption and considering that lattice theory has been extensively explored to deal with problems with aging images, we present in this section the notion of lattice-valued overlap functions.
	
	\begin{definition}\label{overlaps}
		Let $L$ be a bounded lattice. A function $O:L^2 \rightarrow L$ is called a L-overlap  function (simply overlap, if the context is clear) if all of following properties hold:
		\begin{description}
			\item [(OL1)] $O(x,y)=O(y,x)$ for all $x, y \in L$;
			\item [(OL2)] $O(x,y)=0_L$ if and only if $x=0_L$ or $y=0_L$;
			\item [(OL3)] $O(x,y)=1_L$ if and only if $x=y=1_L$;
			\item [(OL4)] $O$ is non-decreasing;
			\item [(OL5)] $O$ is continuous.
		\end{description}
	\end{definition}

\begin{remark}
	Here we are considering the notion of continuity as defined in Section \ref{cont} on lattices but it could be any other one notion that works for lattice framework.
\end{remark}

In some contexts, continuity is not an indispensable property especially when we consider finite lattices as in some cases of digital image processing applications. 

Also, Bustince et al. in \cite{BUSTINCE} justify the continuity of an overlap function  $O$ on $[0,1]$ by saying that requirement is considered in order to avoid $O$ be a uninorm. However it is easy to see that if a uninorm $U$ is an overlap function then $U$ is necessarily a t-norm.

So, differently from what was proposed by Bustince et al. \cite{BUSTINCE} and in the Definition \ref{overlaps}, these reasons lead us to weakening the notion of $L$-overlaps hidding the requirement of $L$-overlaps be continuous.

\begin{definition}\label{quasi_overlaps}
	Let $L$ be a bounded lattice and $O:L^2 \rightarrow L$ be a function. If $O$ satisfies properties (OL1)-(OL4) it is called a quasi-overlap function on $L$, or just quasi-overlap, if the context is clear.
\end{definition}
\begin{remark}
	Obviously all $L$-overlap function is a quasi-overlap function on $L$. When $L$ is finite, by Remark \ref{rem-cont-fin-L},  any quasi-overlap is an overlap.
\end{remark}

\begin{example}
It is easy to see that for any bounded lattice $L$, $O_{\wedge}(x,y)=x\wedge_L y$ is an $L$-overlap function and for each $a\in L/\{0_L,1_L\}$ the function 
\[O_a(x,y)=\left\{\begin{array}{ll}
                   0_L  & \mbox{, if $x=0_L$ or $y=0_L$} \\
                   1_L  & \mbox{, if $x=y=1_L$} \\
                   a    & \mbox{, otherwise}
                  \end{array}
                  \right. \]
                  is a quasi-overlap which  is not an overlap case there is a directed set $\Delta$ such that $\sup \Delta=1_L$ and $1_L\not\in \Delta$ or, equivalently, there is a filtered set $\Delta$ such that $\inf \Delta=0_L$ and $0_L\not\in \Delta$.

\end{example}

%
The notion of divisible t-norms on a lattice presented in Definition \ref{divisible} can be reformulated for the case of (quasi-)overlap functions.

\begin{definition}
A lattice $L$ equipped with some quasi-overlap $O:L^2 \to L$ is called \textit{divisible} if for all $x,y \in L$ with $y\leq_L x$ there exists some $z\in L$ such that $y=O(x,z)$.
\end{definition}	
	
In which follows we present some results regarding to the properties of $L$-overlap functions.

	\begin{proposition}\label{prop3.1}
		Let $L$ be a bounded lattice and  $O$  be a quasi-overlap divisible on $L$. Then
\begin{enumerate}
\item $O$ is associative, i.e. it satifies 
\begin{equation}
O(x,O(y,z)) = O(O(x,y),z), \ \forall  x, y, z \in L
\end{equation} if and only if $O$ is a  positive t-norm;
\item  If  $e \in L$ is a neutral element of $O$ i.e. 
\begin{equation}
O(x,e) = O(e,x) = x, \  \forall x \in L
\end{equation} then  $e=1_L$.
\end{enumerate}
	\end{proposition}
	
	\begin{proof}
\begin{enumerate}
\item $(\Rightarrow)$ Suppose $O$ is an associative quasi-overlap function on $L$.  Since $O$  satisfies commutativity and is non-decreasing, for $O$ be a  positive t-norm, we have only to prove that $1_L$ is neutral element of $O$. Since $O$ is divisible on $L$ and $x \leq_L 1_L$ for all $x \in L$, there exists a $y \in L$ such that $O(y,1_L) = x$. On the other hand, since $O(0_L,1_L)=0_L$ and $O(1_L,1_L)=1_L$, by (OL2) and (OL3), respectively. Then by associativity of $O$ we have: $$O(x,1_L)= O(O(y,1_L),1_L) =O(y,O(1_L,1_L))=O(y,1_L)=x.$$ Similarly it is proved that $O(1,x)=x$. Thus $1_L$ is a neutral element of $O$. 

$(\Leftarrow)$ Conversely, assuming $O$ is a  positive t-norm, we have 
$1_L = O(x,y) \leq_L \min (x,y)$, what gives us $x=y=1_L$. Similary we have $O(1_L,1_L)=1_L$, since $1_L$ is the neutral element of $O$.   Moreover, if $x=0_L$ or $y=0_L$, then of corse that $O(x,y)=0_L$. On the other hand, if it were $0_L <_L x,y$ then since $O$ is a positive t-norm we would have $O(x,y)>_L 0_L$. Therefore, $O(x,y)=0_L \Leftrightarrow x= 0_L \,\textrm{or}\, y=0_L$.
\item  Indeed, if $O(x,e)=O(e,x)=x$ for all $x \in L$ then $O(1_L,e) = 1_L$ and hence by (OL3) $e=1_L$. \end{enumerate} 	
\end{proof}

	
	
%


\subsection{Generalized convex sum of quasi-overlap and overlap functions}


In general the algebraic structure of a lattice does not provide a sum and product operations. However it is possible to generalize the notion of convex sum of overlaps functions given in \cite{BUSTINCE, BEDREGAL} by means family of weight functions as defined by  Lizasoain and Moreno in \cite{Lizasoain}.

\begin{definition}[weight vector]
		Let $L$ be a bounded lattice, $\otimes,\oplus:L^2 \to L$  be a t-norm and a t-conorm respectively. Then $(w_1, \ldots, w_n) \in L^n$ is called a weight vector on $\langle L, \oplus , \otimes \rangle$ if 
\begin{equation}
\bigoplus_{i=1}^{n} w_i = 1_L
\end{equation}
In addition, if for all $\lambda \in L$ we have \begin{equation} \label{dist_weigth_vector}
		\lambda \otimes \left(\bigoplus\limits_{i=1} ^n w_i \right)=\bigoplus\limits_{i=1} ^n\left(\lambda \otimes w_i\right)
		\end{equation} then  $(w_1, \ldots, w_n) \in L^n$ is called a distributive  weight vector on $\langle L, \oplus , \otimes \rangle$.
	\end{definition}
	
	\begin{definition}\cite[family of weight functions]{diego}
		Let $\langle L, \oplus , \otimes \rangle$ be an algebra in which $L$ is a bounded lattice, $\otimes:L^2 \to L$ a t-norm and $\oplus:L^2 \to L$ a t-conorm. A finite family of functions $\mathcal{F}=\{f_i:L^m \to  L |\, i=1,2,\ldots , n\}$ is called the family of weight functions if, for each $w\in L^m$, the vector $(f_1(w), \ldots , f_n(w))$ is a vector of weights in $\langle L, \oplus , \otimes \rangle$. In addition, if a vector $(f_1(w), \ldots , f_n(w))$ satisfies Equation (\ref{dist_weigth_vector})  for all $\lambda \in L$ then $\mathcal{F}$  is called a distributive family of weight functions.
	\end{definition}

Next result gives a generalized version of convex sum of quasi-overlap functions.

	\begin{theorem}\label{lemma}
		Let $O_1, \ldots , O_n: L^2 \rightarrow L$ be quasi-overlap functions on a bounded lattice $L$ and $\otimes,\oplus: L^2 \rightarrow L$ be a t-norm and t-conorm respectively, both continuous. If $\mathcal{O_F}= \{O_1, \ldots , O_n\}$ is a family of weight quasi-overlaps then the function $F:L^2 \to L$ given by 
\begin{equation}
F(x,y) = \bigoplus\limits_{i=1} ^n \lambda_i \otimes O_i(x,y)
\end{equation}
is also  a quasi-overlap function, where $\bigoplus\limits_{i=1} ^n \lambda_i=1_L$ for all $\lambda_i \in L$. In addition, if $\mathcal{O_F}$ are $L$-overlaps and $\oplus$ as well as $\otimes$ are continuous then $F$ is also an $L$-overlap function.
	\end{theorem}
	\begin{proof}We verify that $F$ satisfies the conditions of Definition \ref{quasi_overlaps} (and \ref{overlaps}) as follows.
		\begin{description}
			\item [(OL1)] Straightforward from commutativity of functions $O_i$ with $i = 1, 2, \ldots, n$;
			\item [(OL2)] Suppose $F(x,y)=0_L$. Then $ \bigoplus\limits_{i=1} ^n \lambda_i \otimes O_i(x,y)=0_L$ if and only if we have $\lambda_i \otimes O_i(x,y)=0_L$  for each $i = 1, \ldots , n$. Moreover, since $\bigoplus\limits_{i=1} ^n \lambda_i=1_L$ for all $\lambda_i \in L$ it follows there exists $i_0 \in \{1, \ldots , n\}$ such that $\lambda_{i_0} \otimes O_{i_0}(x,y)=0_L$ if only if $O_{i_0}(x,y)=0_L$ if and only if $x=0_L$ or $y=0_L$: 
			\item [(OL3)]  Suppose $F(x,y)=1_L$. Then $ \bigoplus\limits_{i=1} ^n \lambda_i \otimes O_i(x,y)=1_L$ if and only if there exists $i_0$ such that $\lambda_{i_0} \otimes O_{i_0}(x,y)=1_L$ if and only if $\lambda_{i_0} = O_{i_0}(x,y)=1_L$ if and only if $x=y=1_L$;
			\item [(OL4)] $F$ is increasing (therefore non-decreasing) since it is composed by increasing operations  $\oplus$ and $\otimes$.
			\item  [(OL5)] The continuity of $F$ can be obtained immediately from the continuities of $\mathcal{O_F}$,  $\oplus$ and $\otimes$.
	\end{description} 	\end{proof}
	
	\begin{proposition}
		Let $\otimes:L^2 \rightarrow L$ be a t-norm.	If $\psi, \varphi: L \to L$ are increasing $\{0,1\}$-functions  then the mapping  $$O_{\psi,\varphi}(x,y)=\psi(\varphi(x)\otimes \varphi(y))$$ is a quasi-overlap function.  In addition, if $\otimes$, $\psi$, and $\varphi$ continuous then $O_{\psi,\varphi}$ is an $L$-overlap.
	\end{proposition}
	\begin{proof} It follows by a direct verification of the axioms of Definitions \ref{quasi_overlaps} and \ref{quasi_overlaps}. \end{proof}

	\section{Main properties of quasi-overlap functions}
	

This section is devoted to discuss about the main properties of quasi-overlap and overlap functions namely migrativity, homogeneity, idempotency.

\vspace{0.5cm}
	

	\subsection{$(\alpha, A)$-Migrativity}
	
	

The meaning of migrativity for a quasi-overlap function $O:[0,1]^2 \rightarrow [0,1]$ is that this function is invariant with respect to the same factor $\alpha \in [0,1]$ given in both entries, i.e. $O(\alpha x,y) = O(x,\alpha y)$ for all $x, y \in [0,1]$ (see \cite{BUSTINCE, BEDREGAL,BUSTINCE2}). Here we present a generalized definition of migrativity by means aggregation function.
\begin{definition}
		Let $L$ be a bounded lattice and $A:L^2 \to L$ be an aggregation function. For a given $\alpha \in L$ a bivariate operation $F:L^2 \to L$ is called $(\alpha, A)$-migrative if  it satisfies
		\begin{equation}
		F(A(\alpha,x),y)=F(x,A(\alpha,y)), \quad \textrm{for all} \quad x,y \in L.
		\end{equation}
In case $F$ is $(\alpha,A)$-migrative for all $\alpha \in L$ then it is called just $A$-migrative.
	\end{definition}
	
	\begin{proposition} \label{lemma1}
		Let $A:L^2 \rightarrow L$ be a uninorm with neutral element $a \in L$. A function $F:L^2 \rightarrow L$ is $A$-migrative if and only if there exists a function $f:L \rightarrow L$ such that $F(x,y) = f(A(x,y))$ for all $x, y \in L$.
	\end{proposition}
	
	\begin{proof}
		Suppose $F$ is a A-migrative function. Note that if $A(x,y) = A(z,w)$ then
		$F(x,y) = F(A(a,x),y) = F(a,A(x,y)) = F(a,A(z,w))= F(A(a,z),w) = F(z,w)$. Also, since $a \in L$ is a neutral element of $A$, for every $z \in L$ it follows that $z=A(a,z)$ and hence the function  $f:L \rightarrow L$ given by $f(z) = F(x,y)$ such that $A(x,y) = z$ is a well and univocally defined function which satisfies $F(x,y) = f(A(x,y))$ for all $x, y \in L$. Indeed, since $F$ is a A-migrative function, we have \begin{eqnarray*}
		F(a,z)&=&F(a,A(x,y))\\
		&=&F(A(a,x),y)\\
		&=&F(x,y).
		\end{eqnarray*}
		
		Reciprocally, suppose there exists $f:L \rightarrow L$ such that $F(x,y) = f(A(x,y))$ for all $x, y \in L$. Hence, for all $x, y, \alpha \in L$ it follows 
		
		$$
		\begin{array}{rcll}
		F(A(\alpha,x),y) & = & f(A(A(\alpha,x),y)) & \\
		& = & f(A(A(x,\alpha),y)), & by \ commutativity \ of \ A\\
		& = & f(A(x,A(\alpha,y))), & by \ associativity \ of \ A\\
		& = & F(x,A(\alpha,y))
		\end{array}
		$$ Therefore $F$ is a A-migrative function. \end{proof}
	\begin{remark}
		In particular, if $a=1_L$ then $A$ is a t-norm. The next proposition uses this fact.
	\end{remark}
	\begin{proposition} \label{prop1}
		Under conditions of Proposition \ref{lemma1}, if $A$ is a t-norm and $F:L^2 \rightarrow L$ is $A$-migrative then 
		\begin{enumerate}
			\item $F$ is symmetric;
			\item $F(1_L,1_L) = 1_L$ if and only if $f(1_L) = 1_L$;
			\item $F(0_L,0_L) = 0_L$ if and only if $f(0_L) = 0_L$;
			\item $F$ is continuous if and only if $f$ and $A$ are continuous.
		\end{enumerate}
	\end{proposition}
	
	\begin{proof} 
		\begin{enumerate}
			\item Since $A$ is a t-norm, if $F$ is $A$-migrative then $F(x,y) = F(A(1_L,x),y) = F(1_L,A(x,y)) = F(1_L,A(y,x)) = F(A(1_L,y),x) = F(y,x)$ for all $x,y \in L$;
			\item Notice that $F(1_L,1_L) = 1_L$ if and only if $f(1_L) = f(A(1_L,1_L)) = 1_L$;
			\item Analagous to item 2;
			\item $(\Rightarrow)$ If $F$ is continuous, we must show that $ f $ and $ A $ are also continuous. By Proposition \ref{lemma1}, exists a function $f:L \rightarrow L$ such that $F(x,y) = f(A(x,y))$ for all $x, y \in L$. Let $\Delta \subseteq L^2$ be a directed set. We assert that $A(\Delta)=\{z \in L \,|\, z=A(x,y), \, (x,y)\in \Delta\}$ is also directed set. Indeed, since $\Delta$ is directed set, for all $(u,v), (p,q) \in \Delta$, exists $(r,s)\in \Delta$ such that $(u,v)\leq_{L^2} (r,s)$ and $(p,q)\leq_{L^2} (r,s)$. Hence, by monotonicity of t-norm $A$ it follows that $A(u,v)\leq A(r,s)$ and $A(p,q)\leq A(r,s)$. Thus $A(\Delta)$ is a directed set. Moreover, since $L$ is complete, it follows that exists $\sup \Delta$ and $\sup A(\Delta)$ and it is also easy to see that  $A(\sup \Delta)=\sup A(\Delta)$. Therefore $A$ is continuous. Moreover, since $F$ is continuous, we have $$f(\sup A(\Delta))=f(A(\sup \Delta))=F(\sup \Delta)=\sup F(\Delta) = \sup f(A(\Delta)).$$ Therefore $f$ is continuous.\\ $(\Leftarrow)$ If $f$ and $A$ are continuous functions then it follows straightforward that $F$ is continuous. 	\end{enumerate}  \end{proof}
	
	\begin{theorem}\label{theo4.1}
		Let $A:L^2 \rightarrow L$ be an uninorm. A function $O:L^2 \rightarrow L$ is an $A$-migrative quasi-overlap function if and only if $O(x,y) = f(A(x,y))$ holds for some non-decreasing function $f:L \rightarrow L$ such that $f(0_L) = 0_L$ and $f(1_L) = 1_L$. 
	\end{theorem}
	
	\begin{proof} Straightforward from Propositions \ref{lemma1} and  \ref{prop1}.


	\end{proof} 
	
	\begin{theorem}\label{corolariosoma}
		Let $L$ be a bounded lattice, $A:L^2 \rightarrow L$ be an uninorm and $\otimes,\oplus:L^2 \rightarrow L$ be a t-norm and a t-conorm. A function $F: L^2 \to L$ is $A$-migrative if and only if $F$ is given by $$F(x,y)=\bigoplus\limits_{i=1} ^n \lambda_i \otimes F_i(x,y),$$ where $\lambda_i \in L$ such that $\bigoplus\limits_{i=1} ^n \lambda_i = 1_L$ and $\mathcal{F}_A=\{F_i: L^2 \to L \,| F_i \textrm{ is } A\textrm{-migrative}\}$ ($i=1, 2, \ldots, n$) is a finite family of $A$-migrative weight functions.
	\end{theorem}		
			\begin{proof}  Supposing $F$ is $A$-migrative function then taking $\lambda_1 = \lambda_2 = \cdots = \lambda_{n-1} = 0_L$ and $\lambda_n = 1_L$ it follows that
$$F(x,y)= \left\{ \bigoplus\limits_{i=1}^{n-1} 0_L \otimes F_i(x,y)\right\} \oplus (1_L \otimes F(x,y))$$ where $\bigoplus\limits_{i=1} ^n \lambda_i = 1_L$ and $\mathcal{F}_A=\{F\}$.\\			Reciprocally, suppose function $F:L^2 \to L$ is such that $$F(x,y)=\bigoplus\limits_{i=1} ^n \lambda_i \otimes F_i(x,y),$$ where $\lambda_i \in L$ such that $\bigoplus\limits_{i=1} ^n \lambda_i = 1_L$  and $F_i \in \mathcal{F}_A$. Since each $F_i$ is $A$-migrative for all $\alpha \in L$ we have 
		\begin{eqnarray*}
			F(A(\alpha,x),y)&=&\bigoplus\limits_{i=1} ^n \lambda_i \otimes F_i(A(\alpha,x),y)\\
			&=&\bigoplus\limits_{i=1} ^n \lambda_i \otimes F_i(x,A(\alpha,y))\\
			&=&F(x,A(\alpha,y)).
		\end{eqnarray*}	Therefore $F$ is $A$-migrative. \end{proof}
	
	\begin{corollary}
		A function given by $O(x,y)=\bigoplus\limits_{i=1} ^n \lambda_i \otimes O_i(x,y)$ is a $A$-migrative quasi-overlap function if, and only if, $O_i$ belongs to a finite family of weight quasi-overlaps $A$-migratives for all $x,y, \lambda_i \in L$ such that $\bigoplus\limits_{i=1} ^n \lambda_i=1_L$ and $i=1, 2, \ldots, n$.
	\end{corollary}
	
	
	\subsection{Extended homogeneity}
	
	
	\noindent
	
Recall that a function $F:[0,1]^n \to [0,1]$ is called an homogeneous function of order $k \in \mathbb{N}$ (or simply $k$-homogeneous) if, for any $\lambda \in [0,\infty[$ and $x_i \in [0,1]$, $i \in \{1,\ldots, n\}$, such that $\lambda x_i \in [0,1]$, it holds that 
\begin{equation} \label{restrita}
F(\lambda x_1, \ldots , \lambda x_n)=\lambda^k F(x_1, \ldots , x_n).
\end{equation}
For instance, the $n$-dimensional product given by 
\begin{equation} \label{product}
\prod\limits_{i=1} ^n x_i = \Pi (x_1, \ldots , x_n) = x_1 \cdot x_2 \cdot \ldots \cdot x_{n-1}\cdot x_n
\end{equation} is an homogeneous function of order $n$.
	
In this section we intend to extend the concept of homogeneous functions for lattice-valued overlap functions in order to give a characterization of those kind of functions by means of the notion of power of bivariate functions \cite{DIMURO}.
	
	\begin{definition}\label{grau_k}
		Let $L$ be a bounded lattice and $f:L^2 \to L$ be a function. The power notation $\lambda_f ^{(n)}$, where $n \in \mathbb{N}$, is defined as: \begin{eqnarray} 
			\lambda_f ^{(0)}&=&1_L \nonumber\\
		    \lambda_f ^{(1)}&=&\lambda \nonumber\\
			\lambda_f ^{(n)}&=&f(\lambda, \lambda_f ^{(n-1)}) \label{expoente},
		\end{eqnarray} for all $\lambda \in L$. 
	\end{definition}
	
	\begin{proposition}\label{prop41}
	If $f:L^2 \rightarrow L$ is an associative function and $1_L$ is its neutral element, then $\lambda_f ^{(p+q)}=f(\lambda_f ^{(p)} , \lambda_f ^{(q)})$ for all $p,q \in \mathbb{N}$ and $\lambda \in L$.
	\end{proposition}
	
	\begin{proof}
		Fixed $q> 0$, the demonstration follows by induction on $p$. 	\end{proof}
	
	\begin{proposition}\label{expo}
		Let $\otimes:L^2 \rightarrow L$ be a strict t-norm\footnote{We recall that a t-norm $T$ is said to be strict, if $T$ is continuous and strictly monotone, i.e., $T(x,y)<_LT(x,z)$ whenever $0<_L x$ and $y<_L z$} on a lattice $L$. Then \linebreak $\lambda_{\otimes} ^{(p)} \leq_L \lambda_{\otimes} ^{(q)} \Leftrightarrow p \geq q$, for all $0_L<_L \lambda <_L 1_L$ and $p,q \in \mathbb{N}$.
	\end{proposition}
	
	\begin{proof}
			$(\Rightarrow)$ We will prove by contraposition. If $p< q$ then because, $\otimes$ is strict and $0_L<_L \lambda <_L 1_L$, we have that $0_L <_L \lambda_{\otimes} ^{(k+1)} <_L \lambda_{\otimes} ^{(k)}$ for each $k\in\mathbb{N}$. Therefore $\lambda_{\otimes} ^{(q)} <_L \lambda_{\otimes} ^{(p)}$. \\ 
			$(\Leftarrow)$ If $p> q$ then because, $\otimes$ is strict and $0_L<_L \lambda <_L 1_L$, we have that $0_L <_L \lambda_{\otimes} ^{(k+1)} <_L \lambda_{\otimes} ^{(k)}$ for each $k\in\mathbb{N}$. Therefore $\lambda_{\otimes} ^{(p)} <_L \lambda_{\otimes} ^{(q)}$. 
	
	\end{proof}
	
	\begin{remark}
	Notice that the power notation $\lambda_{\otimes} ^{(n)}$ can be seen as the particular case of a non-increasing net $\left(a_n\right)_{n\in \mathbb{N}}$ whose general term is $a_n = \lambda_{\otimes} ^{(n)}$ for all $\lambda \in L$.
	\end{remark}
	
	\begin{definition}[Extended homogeneity]\label{general}
		 Let $L$ be a bounded lattice and \linebreak$f: L^2 \rightarrow L$ be a function and $k\in \mathbb{N}^*$. A function $F:L^n \to L$ is called an homogeneous extension of order k with respect to $f$ (or just $f^k$-homogeneous) if \begin{equation}
		F(f(\lambda,x_1), \ldots ,f(\lambda,x_n)) = f(\lambda_f ^{(k)},F(x_1,\ldots , x_n)) \label{14}
		\end{equation} holds for all $\lambda , x_1 , \ldots , x_n \in L$.
	\end{definition}
	
	\begin{remark}
		The Definition \ref{general} generalizes the classical notion of homogeneity of order $k$ (cf. Identity (\ref{restrita})). In fact, when $L=[0,1]$ and $f$ is the 2-dimensional product as defined in (\ref{product}) it is clear to see that for any $k$-homogeneous function $F:[0,1]^n \to [0,1]$ we have 
		\begin{eqnarray*}
			F(\Pi(\lambda, x_1), \ldots , \Pi(\lambda, x_n))&=&F(\lambda x_1, \ldots , \lambda x_n)\\
			&=&\lambda^k F(x_1, \ldots , x_n)\\
			&=&\Pi(\lambda_{\Pi} ^{(k)}, F(x_1, \ldots , x_n)), 
		\end{eqnarray*} since that $\lambda_{\Pi} ^{(k)}=\lambda^k$ by induction. Thus $F$ is also $\Pi^k$-homogeneous.
	\end{remark}
	
	\begin{theorem}
		Let $L$ be a bounded lattice, $\rho:L \to L$ be an automorphism  and $f:L^2 \rightarrow L$ be a function such that  \begin{equation} \label{comut}\rho(f(x,y))=f(\rho(x), \rho(y)) \ \forall x, y \in L \end{equation} If $F:L^n \rightarrow L$ is a $f^ k$-homogeneous function then $F^{\rho}$ is also $f^ k$-homogeneous.
	\end{theorem}
\begin{proof}
Notice that $\left( \rho(\lambda) \right) _f ^{(k)}\ = \rho\left(\lambda_f ^{(k)}\right)$ by Definition \ref{grau_k} and Identity (\ref{comut}). Hence, assuming $F$ is $f^k$-homogeneous it follows that 
$$
{\footnotesize \begin{array}{rcll}		
 F^{\rho} \left( f(\lambda,x_1), \ldots , f(\lambda,x_n)\right) & = & \rho^{-1}\left(F(\rho(f(\lambda,x_1)), \ldots , \rho(f(\lambda,x_n)))\right) & \text{by} \ (\ref{eq-act-aut})\\  
& = & \rho^{-1}\left(F(f(\rho(\lambda),\rho(x_1)), \ldots , f(\rho(\lambda),\rho(x_n)))\right) & \text{by} \ (\ref{comut})\\
& = & \rho^{-1} \left( f \left( \left(\rho(\lambda)\right)_f ^{(k)}\right), F\left(\rho(x_1), \ldots , \rho(x_n)\right)\right) & \text{by} \ (\ref{14})\\
& = & f \left(\rho^{-1} \left( \rho\left(\lambda_f ^{(k)}\right)\right) , \rho^{-1} \left( F\left(\rho(x_1), \ldots , \rho(x_n)\right)\right)\right) & \text{by} \ (\ref{auto_comut})\\
& = & f \left(\lambda_f ^{(k)}, F^{\rho}(x_1, \ldots , x_n) \right). & \\
\end{array}}
$$
	\end{proof}
	
	\begin{theorem}
		Let $\otimes$ a t-norm and $\oplus$ a t-conorm, both on bounded lattice $L$. Let $F_i:L^2 \to L$  be a finite distributive family of $\otimes^{k_i}$-homogeneous weight functions ($i = 1, 2, \ldots, n$) and $F:L^2 \to L$ given by $$F(x,y)=\bigoplus\limits_{i=1} ^n w_i \otimes F_i(x,y),$$ where scalar weights $w_i \in L$ are such that $\bigoplus\limits_{i=1} ^n w_i=1_L$. Then $F$ is $\otimes^{k}$-homogeneous if and only if for each $i \in \{1,2,\ldots,n\}$ such that $w_i >_L 0_L$ it holds that $k_i=k$.
	\end{theorem}
	
	\begin{proof} Assume that $F$ is $\otimes^{k}$-homogeneous and consider the set $$I=\{i \in \{1, \ldots , n\}|\, w_i >_L 0_L\}.$$ Then, since each $F_i$ is $\otimes^{k_i}$-homogeneous we have that $$F(\lambda \otimes x, \lambda \otimes y)=\bigoplus\limits_{i\in I} w_i \otimes F_i(\lambda \otimes x,\lambda \otimes y)= \bigoplus\limits_{i\in I} w_i \otimes \left(\lambda_{\otimes} ^{(k_i)} \otimes F_i(x,y)\right)$$ and also $$F(\lambda \otimes x, \lambda \otimes y)=\lambda_{\otimes} ^{(k)} \otimes F(x,y) =\lambda_{\otimes} ^{(k)} \otimes \left( \bigoplus\limits_{i\in I} w_i \otimes F_i(x,y)\right).$$ Therefore, since family of $\otimes^{k_i}$-homogeneous weight functions $\{F_i\}$ is distributive and $\otimes$ is associative, it follows that $$\bigoplus\limits_{i\in I} \left( w_i \otimes \lambda_{\otimes} ^{(k)} \right)\otimes F_i(x,y) = \bigoplus\limits_{i\in I} \left( w_i \otimes \lambda_{\otimes} ^{(k_i)}\right) \otimes F_i(x,y)$$ for each $\lambda \in L$, wich implies that $k=k_i$ for all $i \in I$. \\
	Conversely, assuming $k=k_i$ for all $i \in I$ we have \begin{eqnarray*} F(\lambda \otimes x, \lambda \otimes y)&=&\bigoplus\limits_{i\in I} w_i \otimes F_i(\lambda \otimes x,\lambda \otimes y)\\
			&=& \bigoplus\limits_{i\in I} w_i \otimes \left(\lambda_{\otimes} ^{(k)} \otimes F_i(x,y)\right)\\
			&=& \lambda_{\otimes} ^{(k)} \otimes \left( \bigoplus\limits_{i\in I} w_i \otimes F_i(x,y)\right)\\
			&=&\lambda_{\otimes} ^{(k)} \otimes F(x,y). \end{eqnarray*} Therefore $F$ is $\otimes^{k}$-homogeneous. \end{proof}
	
		\begin{theorem} \label{const_overlap}
	Let $\otimes: L^2 \to L$ be a ($\wedge$ and $\vee$)-distributive t-norm on a bounded lattice $L$ such that the pair $\langle L, \otimes \rangle$ is divisible. A function $F:L^2 \to L$ is $\otimes ^k$-homogeneous, has $1_L$ as neutral element and satisfies $F(x,y)=F(x\wedge y,x\vee y)$ for all $x,y \in L$ if and only if 
		\begin{equation}\label{nonoverlap}
	F(x,y)= (x\wedge y) \otimes \left(x\vee y\right)_{\otimes} ^{(k-1)}, \, \mbox{ for all } x, y \in L.	
		\end{equation}
	\end{theorem}
	\begin{proof}
		$(\Rightarrow)$ Suppose $F(x,y)= F(x\wedge y, x\vee y)$ for all $x,y\in L$. Since the pair $\langle L, \otimes \rangle$ is divisible and $x \wedge y \leq_L x\vee y$, exists $m \in L$ such that \linebreak $(x\vee y) \otimes m = x\wedge y$. Therefore, since $F$ is $\otimes^k$-homogeneous with neutral element $1_L$ one has that \begin{eqnarray*}
			F(x,y)&=&F(x\wedge y, x\vee y)\\
			&=&  F\left(m\otimes (x\vee y), x\vee y\right)\\
			&=& \left(x\vee y\right)_{\otimes} ^{(k)} \otimes F(m,1_L)\\
			&=&m\otimes \left(x \vee y\right)_{\otimes} ^{(k)}\\
			&=&(x\wedge y) \otimes \left(x\vee y\right)_{\otimes} ^{(k-1)}.\end{eqnarray*}
		
		\noindent
		$(\Leftarrow)$ Consider $F$ as defined in Equation (\ref{nonoverlap}). $1_L$ is the neutral element of $F$ since for all $x \in L$ and $k \in \mathbb{N}^*$ we have $x\wedge 1_L = x$, $x\vee 1_L = 1_L$ and $(1_L)_{\otimes} ^{(k-1)} = 1_L$ which implies that $F(1_L,x)=F(x,1_L)=x \otimes 1_L = x$. Moreover, if $r=x\wedge y$ and $s=x\vee y$ then, by Equation (\ref{nonoverlap}), \begin{equation*}
		F(r,s)=(r\wedge s) \otimes (r\vee s)_{\otimes} ^{(k-1)}=r\otimes s_{\otimes} ^{(k-1)}.		
		\end{equation*} Therefore $F(x,y)=F(x\wedge y,x\vee y)$ for all $x,y\in L$. Now, if $x\coh y$,  without loss of generality we assume that $x\leq_L y$ and hence  $\lambda \otimes x\leq_L \lambda \otimes y$ for all $\lambda \in L$. Then, it holds that  \begin{eqnarray*} 
			F \left(\lambda \otimes x, \lambda \otimes y \right) &=& (\lambda \otimes x) \otimes \left( \lambda \otimes y\right)_{\otimes} ^{(k-1)}\\
			&=&(\lambda)_{\otimes} ^{(k)} \otimes \left( x \otimes  (y)_{\otimes} ^{(k-1)} \right)\\ 
			&=&(\lambda)_\otimes ^{(k)} \otimes ((x\wedge y)\otimes (x\vee y)^{(k-1)}_\otimes)\\
			&=& (\lambda)_{\otimes} ^{(k)} \otimes F(x,y).
		\end{eqnarray*} Finally, if $x \parallel y$ we must show that $F(\lambda \otimes x, \lambda \otimes y)= \lambda_{\otimes} ^{(k)} \otimes F(x,y), \, \forall \lambda \in L$. So, by Equation (\ref{nonoverlap}), \begin{equation}\label{equasup}
	F(\lambda \otimes x,\lambda \otimes y)=\left[(\lambda \otimes x)\wedge (\lambda \otimes y)\right] \otimes \left[(\lambda \otimes x)\vee (\lambda \otimes y)\right]_{\otimes} ^{(k-1)}.
\end{equation} Moreover, since $\otimes$ is ($\wedge$ and $\vee$)-distributive, we have \begin{eqnarray*}
(\lambda \otimes x)\wedge (\lambda \otimes y)&=&\lambda \otimes(x\wedge y)\\
(\lambda \otimes x)\vee (\lambda \otimes y)&=&\lambda \otimes(x\vee y).
\end{eqnarray*} Therefore the Equation (\ref{equasup}) can be rewritten as \begin{eqnarray*}
F(\lambda \otimes x,\lambda \otimes y)&=& \left[\lambda \otimes (x\wedge y)\right] \otimes \left[\lambda \otimes (x\vee y)\right]_{\otimes} ^{(k-1)}\\
&=& \lambda_{\otimes} ^{(k)} \otimes (x\wedge y) \otimes \left(x\vee y\right)_{\otimes} ^{(k-1)}\\
&=& \lambda_{\otimes} ^{(k)} \otimes  F(x,y).
\end{eqnarray*} Therefore $F$ is $\otimes ^k$-homogeneous. \end{proof}
	
	\begin{remark}\label{obs4.4}
		Notice that the $F$ function defined in Equation (\ref{nonoverlap})  not necessarilly is a  quasi-overlap function since axiom \textbf{(OL2)} of Definition \ref{overlaps} can fail when $x \parallel y$.
	\end{remark}
	
	\begin{corollary}\label{col4.2}
		The function $F$ as defined in Equation (\ref{nonoverlap}) is a quasi-overlap function if and only if $L$ is a chain.
	\end{corollary}
	
	\begin{proof}
		Straightforward from Theorem \ref{const_overlap} and Definition \ref{overlaps}. \end{proof}

	\begin{theorem}\label{teo4.6}
		Let $\otimes: L^2 \to L$ be a ($\wedge$ and $\vee$)-distributive t-norm on a bounded lattice $L$ such that the pair $\langle L, \otimes \rangle$ is divisible and $F_1, F_2 :L^2 \to L$ be $\otimes$-homogeneous functions of order $k_1$ and $k_2$, respectively, such that $F_1$ and $F_2$ have $1_L$ as neutral element and satisfy $F_i(x,y)=F_i(x\wedge y,x\vee y)$, $i \in \{1,2\}$, for all $x,y \in L$. Under these conditions, $F_1 \leq_{L} F_2$ if, and only if, $k_1 \geq k_2$.
	\end{theorem}
	\begin{proof}		
		$(\Rightarrow)$ Notice that for all $\lambda \in L\backslash \{0_L,1_L\}$ we have that $\lambda_{\otimes} ^{(k_1)}, \lambda_{\otimes} ^{(k_2)} \in L\backslash \{0_L,1_L\}$. Moreover, by Equation (\ref{14}): $$F_1\left(\lambda \otimes x, \lambda \otimes y\right)=\lambda_{\otimes} ^{(k_1)}\otimes \,F_1 (x,y) \textrm{ and } F_2\left(\lambda \otimes x, \lambda \otimes y\right)=\lambda_{\otimes} ^{(k_2)}\otimes \, F_2 (x,y).$$ Then, since for each $i=1,2$ the function $F_i:L^2 \rightarrow L$ is $\otimes^{k_i}$-homogeneous, by Theorem \ref{const_overlap}, if $F_1 \leq_{L} F_2$, by Equation (\ref{nonoverlap}), in case $x=y=1_L$, we have that $\lambda_{\otimes} ^{(k_1)} \leq_L\lambda_{\otimes} ^{(k_2)}$ and hence by Proposition \ref{expo}, $k_1 \geq k_2$.\\
		
		\noindent			
		$(\Leftarrow)$ By Theorem \ref{const_overlap} we must consider two cases:
		\begin{enumerate}
			
			\item [Case 1.]	If $x = y = 0_L$ then $F_1(x,y)=0_L\leq_L F_2(x,y)=0_L$. If $0<_L x,y \mbox{ and } \linebreak x\coh y$ then    without loss of generality we can consider that $y\leq_L x$. In this case, since $k_1 \geq k_2$, by Proposition \ref{expo} we have $y_{\otimes} ^{(k_1)} \leq_L y_{\otimes} ^{(k_2)}$ and hence $y_{\otimes} ^{(k_1)} \otimes z \leq y_{\otimes} ^{(k_2)} \otimes z$ for all $y, z \in L$. Then since $\langle L, \otimes \rangle$ is divisible there exists $z\in L$ such that $y \otimes z=x$. Therefore, since $F_i$ is ${\otimes} ^{k_i}$ homogeneous for $i=1,2$ then \begin{eqnarray*}
					F_1(x,y)&=& F_1(y \otimes z,y\otimes 1_L)\\ &=& y_{\otimes} ^{(k_1)} \otimes F_1(z,1_L)\\ &=& y_{\otimes} ^{(k_1)} \otimes z\\&\leq_L& y_{\otimes} ^{(k_2)} \otimes z\\
					&=& y_{\otimes} ^{(k_2)} \otimes F_1(z,1_L)\\&=& F_2(y \otimes z,y\otimes 1_L)\\&=&F_2(x,y).
			\end{eqnarray*}	
		\item [Case 2.]	$x \parallel y$. In this case since the pair $\langle L, \otimes \rangle$ is divisible there exists $m\in L$ such that $(x\vee y) \otimes m=x\wedge y$. Therefore, for $i\in \{1,2\}$, $F_i(x,y)=F_i\left(x\wedge y,x\vee y\right)= F_i\left((x\vee y)\otimes m,x\vee y\right)$.  From this point forward the reasoning is analogous to the previous case.\end{enumerate}\end{proof}
	
	According to the Corollary \ref{col4.2}, when $ L $ is a chain each $ F_i $ $( i = 1,2 )$, is an overlap function. So we have the following result.		
	
	\begin{corollary} 
	Let $\otimes: L^2 \to L$ be a t-norm on a chain $L$ such that the pair $\langle L, \otimes \rangle$ is divisible and let $O_1, O_2 :L^2 \to L$ be $\otimes$-homogeneous quasi-overlap functions of order $k_1$ and $k_2$, respectively.  Then, it holds that:
		\begin{enumerate}
			\item [(i)]	If $O_1 \leq_L O_2$ then $k_1 \geq k_2$;
			\item [(ii)] Whenever $O_1$ and $O_2$ have $1_L$ as neutral element, if $k_1 \geq k_2$ then $O_1 \leq_L O_2$.
		\end{enumerate}
	\end{corollary}
	
	\begin{proof}
		Straightforward from Theorem \ref{teo4.6}.
	\end{proof}

	\subsection{Idempotency}

	
	Recall that an element $a\in L$ is called an \textit{idempotent} element of a function $f:L^2\to L$ if $f(a,a)=a$. In the case of every element $a\in L$ is an idempotent element of $f$ the $f$ is called an idempotent function. Note that $0_L$ and $1_L$ are trivial idempotent elements for any quasi-overlap function $O$.


\begin{proposition}
Let $O:L^2 \to L$ be an overlap function. If  we have $\lim\limits_{x \to a^+}O(x,x)=a$ for some $a \in L\backslash \{0_L,1_L\}$ then $a$ is an idempotent element of $O$.
\end{proposition}
\begin{proof}
If $\lim\limits_{x \to a^+}O(x,x)=a$ then by Definition \ref{limitelateral} we have that $a$ is an eventual lower bounded of a net $(x_i)_{i\in J}$ in $L$ and $\lim_{i\in J} x_i=a$ implies $\lim_{i\in J} O(x_i,x_i)=a$. Moreover, since $O$ is continuous  one has that \begin{eqnarray*}
O(a,a)&=&O\left(\lim_{i\in J}x_i, \lim_{i\in J}x_i\right)\\
&=&\lim_{i\in J} O(x_i,x_i)\\
&=&a.\end{eqnarray*}\end{proof}
\begin{proposition}\label{idem}
Let $O:L^2 \to L$ be a quasi-overlap function and $a\in L$. If $a$ is an idempotent element of $O$ then there exists $x \in L$ such that $a= \lim\limits_{n \to \infty}x_O ^{(n)}$, i.e, $a$ is an eventual upper bound of a net $(x_O ^{(n)})_{n \in \mathbb{N}}$.
\end{proposition}
\begin{proof}
If $a$ is an idempotent element of $O$ then $a_O ^{(n)}=a$ for all $n\in \mathbb{N}$ and hence $a= \lim\limits_{n \to \infty}a_O ^{(n)}$.\end{proof}

	\subsection{Cancellation law}


\begin{definition}
A quasi-overlap function $O:L^2 \to L$ satisfies the cancellation law if $O(x,y)=O(x,z)$ implies that $x=0_L$ or $y=z$. In this case, $O$ is called a cancellative quasi-overlap.
\end{definition}

\begin{example}
Let $L = \langle [0,1], \leq \rangle$ be a bounded lattice. Function $O_{DB}:L^2 \rightarrow L$ given by 
$$
O_{DB}(x,y) = \left\{
\begin{array}{rl}
\displaystyle \frac{2xy}{x+y}, & \ if \ x+y \neq 0;\\
0, & \ otherwise.
\end{array}
\right.
$$ is an overlap that satisfies the cancellation law (Example 4.5 in \cite{DIMURO}). 
\end{example}

\begin{theorem} \label{cancellative}
If	a quasi-overlap function $O:L^2 \to L$  is cancellative then it is strictly increasing, i.e. $O(x,y)<_L O(x,z)$ whenever $y <_L z$ and $0_L <_L x$.
\end{theorem}
\begin{proof}
Suppose that $y <_L z$, $0_L <_L x$ and that $O$ is cancellative. By \textbf{(OL4)}, one has that $O(x,y)\leq_L O(x,z)$. Consider $O(x,y)= O(x,z)$. Then, since $O$ is cancellative, $x=0_L$ or $y=z$, which is a contradiction. Thus, one concludes that $O(x,y)<_L O(x,z)$.
\end{proof}
\begin{example}
	Let be $L$ the bounded lattices as in Figure \ref{fig2}. Then  $O_L:{L}^{2} \rightarrow L $ given as in Table \ref{tab:1} is an overlap function that do not satisfy the cancellation law since it not strictly increasing. For instance, $O_L(b,c)=O_L(b,d)$ however $b\neq0_L$ and $c\neq d$. Actually, there is no cancellative overlaps on finite bounded lattices as one can see on Corollary \ref{pigeon}. 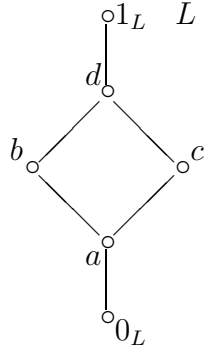
\begin{figure}[!htp]
		\setlength{\unitlength}{1cm}
		\begin{picture}(5,5.5)
		\put(7.5,4){\bf $L$} \put(6.5,4){$\circ$} \put(6.5,3){$\circ$} \put(5.5,2){$\circ$}
		\put(7.5,2){$\circ$} \put(6.5,1){$\circ$} \put(6.5,0){$\circ$} \put(6.7,4){$1_L$} \put(6.3,3.2){$d$}
		\put(5.3,2.2){$b$} \put(7.7,2.2){$c$} \put(6.3,0.8){$a$} \put(6.7,-0.2){$0_L$}
		\put(6.58,3.2){\line(0,1){0.8}} \put(6.49,1.19){\line(-1,1){0.8}} \put(6.49,2.97){\line(-1,-1){0.8}}
		\put(6.68,1.19){\line(1,1){0.8}} \put(6.68,2.97){\line(1,-1){0.8}}\put(6.58,0.2){\line(0,1){0.8}}
		\end{picture}
		\caption{Hasse diagram of lattice $L$} \label{fig2}
	\end{figure}
		
	\begin{table}[htb]
		\centering
		\begin{tabular}{|c||c|c|c|c|c|c|}
			\hline
			$O_L(\cdot,\cdot)$ & $0_L$ & $a$ & $b$ & $c$ & $d$ & $1_L$ \\ \hline \hline
			$0_L$ & $0_L$ & $0_L$ & $0_L$ & $0_L$ & $0_L$ & $0_L$ \\ \hline
			$a$ & $0_L$ & $a$ & $b$ & $c$ & $d$ & $d$\\ \hline
			$b$ & $0_L$ & $b$ & $c$ & $d$ & $d$ & $d$\\ \hline
			$c$ & $0_L$ & $c$ & $d$ & $d$ & $d$ & $d$\\ \hline
			$d$ & $0_L$ & $d$ & $d$ & $d$ & $d$ & $d$\\ \hline
			$1_L$ & $0_L$ & $d$ & $d$ & $d$ & $d$ & $1_L$\\ \hline
		\end{tabular}
		\caption{Tables of overlap function $O_L$.}
		\label{tab:1}
	\end{table}
\end{example}

\begin{corollary} \label{pigeon}
	There is no cancellative quasi-overlap function on finite bounded lattices.
\end{corollary}
\begin{proof}
	Suppose $L$ is a finite bounded lattice and $O$ is a cancellative quasi-overlap function on $L$. Then by Theorem \ref{cancellative} the quasi-overlap $O$ is strictly increasing, i.e. $O(x,y)<_L O(x,z)$ whenever $y <_L z$ and $0_L <_L x$ what implies that $O$ restricted to $\{x\} \times L$ (for a given $x \in L$) should be an injective function over $L\backslash \{1_L\}$ what is contradiction with the pigeonhole principle.
\end{proof}

It is known from the literature that strictly monotonic and cancellation are equivalents properties for overlap functions on the unit interval with the standard linear order (see \cite{DIMURO}). However, the next example reveals that this is not true for $L$-overlap functions.

\begin{example}
 Let $\alpha$ be  a real number such that $0< \alpha < 1$ and  $L=[0,1]\cup\{\alpha\}$ the lattice with usual order when restricted to $[0,1]$ and $x \parallel \alpha$ for all $0<x<1$. Then, the mapping defined by \begin{enumerate}
 	\item [(i)] $O(x,y)= x\cdot y$\, if $x, y \in [0,1]$;
	\item [(ii)] $O(x,\alpha)=O(\alpha,x)=\frac{x}{2}$\, if $x\in [0,1]$ and
	\item [(iii)] $O(\alpha,\alpha)=0.4$
\end{enumerate} is an $L$-overlap function which is strictly increasing but is not cancellative, since $O(0.8,0.5)=0.4 =O(0.8,\alpha)$, but $0.5\parallel \alpha$.
\end{example}

\begin{theorem}\label{strictly}
If a quasi-overlap function $O:L^2 \to L$  is strictly increasing and satisfies $O(x,y\vee z)=O(x,y)\vee O(x,z)$, for all $x,y,z \in L$ then it is cancellative.
\end{theorem}

\begin{proof}
If $O$ is strictly increasing, then $O(x,u)<_L O(x,v)$ whenever $u<_Lv$ and $0_L <_L x$. Suppose that $O$ is not cancellative. Then, there exist $x,y,z \in L$ with $x\neq 0$ such that $O(x,y)= O(x,z)$ and $y \neq z$, i.e either $y<_L z$ or $z<_L y$ or $y \parallel z$. Considering $y <_Lz$  one may conclude that  $O(x,y)<_L O(x,z)$  since $O$ is strictly increasing, which is a contradiction. Similary, the same result is obtained for $z<_L y$. Now consider that $y \parallel z$. Since L is a lattice, it follows that $y\vee z$ exist.  Thus, since $O$ is strictly increasing, one has that: $O(x,y)<_L O(x,y\vee z)$ and $O(x,z)<_L O(x,y\vee z)$. However, $O(x,y) \vee O(x,z) = O(x,y\vee z)$ and this is sufficient to conclude that $O(x,y) \parallel O(x,z)$. Therefore, $O$ is cancellative.
\end{proof}
\begin{corollary}
	Let $L$ be a chain. A quasi-overlap function $O:L^2 \to L$  is cancellative if and only if it is strictly increasing.
\end{corollary}

\begin{proof}
The sufficiency follows from Theorem \ref{cancellative}. To see the necessity just consider the Theorem \ref{strictly} once the condition $O(x,y\vee z)=O(x,y)\vee O(x,z)$ trivially is satisfied by any quasi-overlap function when $L$ is a chain.
\end{proof}

	
	\section{Archimedean quasi-overlap functions and related properties}


From algebraic point of view, the meaning of a set $X$ has the Archimedian Property is that it has no infinitely  large (small) element (other than neutral) . For instance, if $(G,*)$ is a group\footnote{A group $(G,*)$ is a nonempty set $G$ equipped with an operation $*$ which is associative, has neutral element and symmetric element (see \cite{bhattacharya})} then given $x, y \in G$ there exists a $n \in \mathbb{N}^*$ such that $\underbrace{x * \cdots * x}_{n-times} < y$. This concept can naturally be extended for other contexts including for lattices \cite{DIMURO}. Here we discuss about that property for $L$-overlap functions as follows.

	\begin{definition}
		Let $L$ be a bounded lattice. A quasi-overlap function $O: L^2 \to L$ is called Archimedean if for each $x,y \in L \backslash \left\{0_L, 1_L \right\}$ there exists $n \in \mathbb{N}^*$ such that $x_O ^{(n)} <_L y$, where $x_O ^{(n)}$ is given in Equation (\ref{expoente}).
	\end{definition}
	
	\begin{example}
		Let $\otimes :L^2 \to L$ be a strict continuous t-norm. It is easy to see that the function $O_p:L^2 \to L$ given by $O_p(x,y)=x_{\otimes} ^{(p)}\otimes y_{\otimes} ^{(p)}$ with $p>1$ is an overlap function. Since for all $n \in \mathbb{N}^*$ one can verify that  $$x_{O_p} ^{(n)} = x_{\otimes} ^{(2p^{n-1} + p^{n-2}+p^{n-3}+ \ldots + p)}$$ then for all $x,y \in L \backslash \left\{0_L, 1_L \right\}$ it holds that $$\lim\limits_{n \to \infty}x_{O_p} ^{(n)} =\lim\limits_{n \to \infty} x_{\otimes} ^{(2p^{n-1} + p^{n-2}+p^{n-3}+ \ldots + p)} \stackrel{(Prop. \ref{expo})}{=} 0_L <_L y.$$ Therefore $O_{p}$ is an Archimedean overlap function.
	\end{example}
\begin{lemma}\label{basecase}
Let $O: L^2 \to L$ be an Archimedean quasi-overlap function. Then for all $x \in L \backslash \left\{0_L, 1_L \right\}$ it holds that $O(x,x) <_L x$ or $O(x,x) \parallel x$. 
\end{lemma}
\begin{proof}
Since $O$ is Archimedean there exists $n \in \mathbb{N}^*$ such that $x_{O} ^{(n)} <_L x$ and hence $n\neq 1$ since  $x_{O} ^{(1)} = x$.  So taking the least $n \neq 1$ such that $x_{O} ^{(n)} <_L x$ it follows that $x_{O} ^{(n-1)}\geq_L x$ or  $x_{O} ^{(n-1)} \parallel x$. If $x_{O} ^{(n-1)}\geq_L x$ then $O(x,x) \leq_L O\left(x_{O} ^{(n-1)},x\right)=x_{O} ^{(n)}<_L x$. On the other hand, if $x_{O} ^{(n-1)} \parallel x$  then we have the following possibilities:
\begin{enumerate}
	\item[(i)] Suppose $O\left(x_{O} ^{(n-1)},x\right)=x_{O} ^{(n)}$ and  $O(x,x)$ are incomparable. Notice that if  $x \leq_L O(x,x)$ we would have $x_{O} ^{(n)} <_L x \leq_L O(x,x)$ which is contradicts with $x_{O} ^{(n)}  \parallel O(x,x)$. Therefore we must have $O(x,x) <_L x$ or $O(x,x) \parallel x$;
	\item[(ii)] In case $x_{O} ^{(n)}=O\left(x_{O} ^{(n-1)},x\right)\geq_L O(x,x)$ it follows that $O(x,x)\leq_L x_{O} ^{(n)}<_L x$;
	\item[(iii)] Finally, suppose $x_{O} ^{(n)}=O\left(x_{O} ^{(n-1)},x\right)\leq_L O(x,x)$. In this case, due to $x_{O} ^{(n)} <_L x$ we shall have $O(x,x)<_L x$ or $O(x,x) \parallel x$. Indeed, if $x \leq_L O(x,x)$ then applying the function $O$ ($n-2$) times we get the chain: $$x \leq_L x_{O} ^{(2)}\leq_L x_{O} ^{(3)}\leq_L \ldots \leq_L x_{O} ^{(n-1)} \leq_L x_{O} ^{(n)}<_L x,$$ which is obviously a contradicts.
\end{enumerate} \end{proof}

The above result is generalized in the following theorem.

\begin{theorem}\label{teo51}
Let $O: L^2 \to L$ be an Archimedean quasi-overlap function. Then for all $x \in L \backslash \left\{0_L, 1_L \right\}$ it holds that $x_{O} ^{(n+1)} <_L x_{O} ^{(n)}$ or $x_{O} ^{(n+1)}\parallel x_{O} ^{(n)}$. \end{theorem}
\begin{proof}
The proof follows by induction on $n$. In fact, for $n=1$ by Lemma \ref{basecase} we have that $x_{O} ^{(2)}=O(x,x) <_L x = x_{O} ^{(1)}$ or $x_{O} ^{(2)} \parallel x_{O} ^{(1)}$. \\ 
Now, for a given $p\in \mathbb{N}^*$ assume as indution hypothesis that 
$$
\begin{array}{rcll}
x_{O} ^{(p+1)} <_L x_{O} ^{(p)} & \text{or} &  x_{O} ^{(p+1)}\parallel x_{O} ^{(p)} &  \text{(IH)}
\end{array} 
$$
 We shall prove that $x_{O} ^{(p+2)} <_L x_{O} ^{(p+1)}$ or $x_{O} ^{(p+2)}\parallel x_{O} ^{(p+1)}$. Indeed, suppose by absurd that $x_{O} ^{(p+1)} \leq_L x_{O} ^{(p+2)}$. Hence by (IH) if $x_{O} ^{(p+1)} <_L x_{O} ^{(p)}$ due to $O$ is non-decreasing then $x_{O} ^{(p+2)}= O\left(x_{O} ^{(p+1)},x\right) <_L O\left(x_{O} ^{(p)},x\right)=x_{O} ^{(p+1)}$ which is a contradiction with the assumption $x_{O} ^{(p+1)} \leq_L x_{O} ^{(p+2)}$. Otherwise suppose by (IH) we have $x_{O} ^{(p+1)}$ and $x_{O} ^{(p)}$  incomparable. Notice that there is $m\in \mathbb{N}^*$ such that $x_{O} ^{(m)} <_L x$, for all $x \in L \backslash \left\{0_L, 1_L \right\}$  since $O$ is Archimedean.  Thus applying $(p-1)$-times the overlap $O$ we get $x_{O} ^{(p+m-1)} <_L x_O ^{(p)}$. On the other hand by assumption $x_{O} ^{(p+1)} \leq_L x_{O} ^{(p+2)}$ we can get the chain $$x_{O} ^{(p+1)} \leq_L x_{O} ^{(p+2)} \leq_L x_{O} ^{(p+3)}\leq_L \ldots \leq_L x_{O} ^{(p+m-1)}\leq_L \ldots \leq_L 1_L.$$ and hence  $x_{O} ^{(p+1)} \leq_L  x_{O} ^{(p+m-1)} <_L x_O ^{(p)}$ which is a contradiction with assumption $x_{O} ^{(p+1)} \parallel x_{O} ^{(p)}$. Therefore, we must have $x_{O} ^{(p+2)} <_L x_{O} ^{(p+1)}$ or $x_{O} ^{(p+2)}\parallel x_{O} ^{(p+1)}$.\end{proof}



\begin{lemma}\label{6.2}
An Archimedean quasi-overlap function has only trivial idempotent elements.
\end{lemma}
\begin{proof}
Suppose there exists an  idempotent element  $x \in L\backslash\{ 0_L, 1_L\}$ of  a $L$-overlap function $O$. In this case, notice that $x_O ^{(2)}=O(x,x)=x$, $x_O ^{(3)}=O(x,x_O ^{(2)})=O(x,x)=x$ and hence $x_O ^{(n)}=x$ for all $1 < n \in \mathbb{N}$. Then for every $y \in L\backslash \{0_L, 1_L\}$ such that $x >_L y$ it holds that $x_O ^{(n)}=x >_L y$ for all $1 < n \in \mathbb{N}$ which is a contradiction with the fact that $O$ is an Archimedean quasi-overlap function. Therefore $O$ has only trivial idempotent elements.\end{proof}

\begin{definition}
An overlap function $O: L^2 \to L$ has the \textit{limiting} property  if  $\lim\limits_{n \to \infty} x_O ^{(n)}=0_L$ for all $x \in L\backslash \{0_L, 1_L\}$.
\end{definition}
\begin{theorem}\label{t6.1}
Let $O:L^2 \to L$ be an overlap function and consider the following statements:
\begin{enumerate}
	\item [(i)] $O$ satisfies limiting property;
	\item [(ii)] $O$ is Archimedean;
	\item [(iii)] $O$ has only trivial idempotent elements and there exists $b\in L\backslash\{0_L,1_L\}$ such that $O(b,b)=a$ whenever $\lim\limits_{x \to a^+}O(x,x)=a$ for some $a \in L\backslash\{0_L,1_L\}$. 
\end{enumerate} Then it holds that $(i) \Rightarrow (ii)$, $(ii) \Rightarrow (iii)$ and $(i) \Rightarrow (iii)$.
\end{theorem}
\begin{proof}
$(i) \Rightarrow (ii):$ If $O$ satisfies the limiting property then for all  $x \in L\backslash\{0_L,1_L\}$ it holds that $\lim\limits_{n \to \infty} x_O ^{(n)}=0_L$. Therefore for all  $x \in L\backslash\{0_L,1_L\}$ there exists $n \in \mathbb{N}^*$ such that $x_O ^{(n)}<_Ly$.\\
$(ii) \Rightarrow (iii):$ If $O$ is Archimedean then by Lemma \ref{6.2} it has only trivial idempotent elements. Now, consider that $\lim\limits_{x \to a^+}O(x,x)=a$ for some $a \in L\backslash \{0_L,1_L\}$ and $O(y,y) >_L a$ or $O(y,y) \parallel a$ for all $y \in L\backslash \{0_L,1_L\}$. Then for all $y_1, y_2 \in  L\backslash \{0_L,1_L\}$ we have two possibilities:\begin{enumerate}
\item [(1)] $y_1$ and $y_2$ are comparable. In this case, we can assume without loss of generality that $y_1 \leq_L y_2$ and hence $O(y_1,y_2)\geq_L O(y_1,y_1)>_L a$ or $O(y_1,y_2)\geq_L O(y_1,y_1)$ but $O(y_1,y_2) \parallel a$. Thus for all $y \in L\backslash \{0_L,1_L\}$ it holds that $y_O ^{(2)}=O(y,y)>_L a$ or $y_O ^{(2)} \parallel a$. Now assume that $y_O ^{(n)}>_L a$ or $y_O ^{(n)} \parallel a$ for some $1 < n \in \mathbb{N}$. Then, since that $y_O ^{(n)} <_L y$ or $y_O ^{(n)} \parallel y$ (by Theorem \ref{teo51}), it holds that $y_O ^{(n+1)}=O(y,y_O ^{(n)})\geq_L O(y_O ^{(n)},y_O ^{(n)})>_L a$ or $y_O ^{(n+1)} \parallel a$. Therefore for all $n \in \mathbb{N}^*$ we can conclude taht $y_O ^{(n)}>_L a$ or $y_O ^{(n)} \parallel a$  which contradicts the fact of $O$ be Archimedean.
\item [(2)]  $y_1\parallel y_2$. In this case we also have two possibilities. The first one is the case where $O(y_2,y_2)$ and $O(y_1,y_1)$ are comparables. Then, the proof is analogous to case (1). The second one is the case where $O(y_1,y_2) \parallel O(y_1,y_1)$.  In this case we shall prove that  $O(y_1,y_2) >_L a$ or $O(y_1,y_2) \parallel a$.  In fact, if  $O(y_1,y_2) \leq_L a$ holds then by Definition \ref{limitelateral} there are nets $(q_r)_{r\in \mathbb{N}}$ and $(q_s)_{s\in\mathbb{N}}$ in $L$ that converge  for $y_1$ and $y_2$  respectively, since that $a$ is eventual lower bounded. Then, exists $k > \max \{s_0,r_0\}$ such that $O(q_k,q_k)\leq_L q_k$. Since that $q_k \in L\backslash\{0_L,1_L\}$ and $O$ has only trivial idempotent elements, it follows by Theorem \ref{teo51} that $O(q_k,q_k)<_L q_k$ implies $a=\lim\limits_{k \to \infty}O(q_k,q_k) <_L \lim\limits_{k \to \infty}q_k=a$ or  $a=\lim\limits_{k \to \infty}O(q_k,q_k)\parallel \lim\limits_{k \to \infty}q_k=a$, which is a contradiction. 
\end{enumerate} Therefore always there exists $b \in L\backslash\{0_L,1_L\}$ such that $O(b,b)=a$, for some $a \in L\backslash\{0_L,1_L\}$. \\ 
\noindent
$(i) \Rightarrow (iii):$ Straightforward. \end{proof}

\section{Final remarks}
In this article, we presented the concept of lattice-valued overlap functions making a wide discussion about the main properties of that operators in order to investigate its potentialities. We also propose the definition of quasi-overlap functions, in the case where the continuity of overlap functions is not indispensable. 

The results showed that in most cases, the properties are naturally to the scope of the lattices and are preserved. It is worth highlighting the property of homogeneity that can be extended by using the structure provided by the families of weight functions concept.

Other properties that deserve attention were discussed in detail in the Proposition \ref{prop3.1} as well as  in the Theorems \ref{const_overlap} and \ref{teo4.6}, where the concepts of divisible quasi-overlap and divisible t-norm on a bounded lattice $L$ were used to replace the known intermediate value theorem (note that these concepts coincide only when $L$ is a chain). Moreover, the additional hypothesis of t-norm being \linebreak($\wedge$ and $\vee$)-distributive on a bounded lattice $L$ can be replaced by any residuated lattice $\langle L, \wedge, \vee, \otimes, \Rightarrow, 0_L,1_L\rangle$.

It is also worth noting that, unlike the overlap functions on the unit interval with the standard linear order, strictly monotonic and cancellation properties do not are equivalents for $L$-overlap functions when $L$ is not a chain. However, if $L$ is an any bounded lattice, a quasi-overlap function strict $O$ is also a cancellative quasi-overlap function when we add the hypothesis $O (x, y \vee z) = O (x, y) \vee O (x, z)$, for all $x,y,z \in L$. In other words, a strict quasi-overlap is an cancellative quasi-overlap if, and only if, the structure $\langle L, \leq_L, O, 1_L \rangle$ is an integral commutative groupoid with neutral element $1_L$ and satisfying $O (x, y \vee z) = O (x, y) \vee O (x, z)$, for all $x,y,z \in L$.

As future works, obviously this paper can be continued in several ways, but some of them seem to us of immediate interest. On the one
hand, we can search for alternative characterizations for $L$-overlap functions, specifically designed we want to deepen the respect of some classes of overlapping functions, besides their characterization via homomorphisms, as well as the investigation of interval-valued of quasi-overlap and overlap functions. And on the other hand, we can explore additional properties involving the residuation of $L$-overlap functions.

\section*{References}
\bibliographystyle{elsarticle-num} 
\bibliography{bib}

%
%
%
\end{document}